%% file: strings.tex
\documentclass{article}[11pt]

\usepackage{times}

\usepackage{graphicx}
\usepackage{amssymb}
\usepackage{color}
\usepackage{graphics}
\usepackage{epsfig}
\usepackage{latexsym}
\usepackage{amsmath}
\usepackage{fullpage}

\usepackage{graphicx}

\usepackage{subfigure}

\newcommand {\noop}[1]{}

\include{macros}

\begin{document}


\title{On the Computational Complexity \\of Satisfiability Solving for String Theories}


\author{Susmit Jha \qquad Sanjit A. Seshia \qquad Rhishikesh Limaye\\
EECS Department, UC Berkeley\\
\{jha,sseshia,rhishi\}@eecs.berkeley.edu}

\maketitle

\begin{abstract} 
Satisfiability solvers are increasingly playing a key role in 
software verification, with particularly effective use 
in the analysis of security vulnerabilities.
String processing is a key part of many software applications,
such as browsers and web servers.
These applications are susceptible to attacks through malicious data 
received over network. Automated tools for analyzing the security
of such applications, thus need to reason about strings.
For efficiency reasons, it is desirable to have a solver that treats
strings as first-class types. In this paper, we present some theories
of strings that are useful in a software security context and analyze the
computational complexity of the presented theories. We use this 
complexity analysis to motivate a byte-blast approach which employs a
Boolean encoding of the string constraints to a corresponding 
Boolean satisfiability problem.
\end{abstract}


\section{Introduction}

Many security-critical applications such as Web servers routinely
process strings as an essential part of their functionality. They take
strings as inputs, screen them using filters, manipulate them and use
them for operations such as database queries. It is pertinent to
verify that these programs do not have vulnerabilities which can be
used to compromise system security. Verification and structured
testing techniques to validate security of such applications often
rely on using constraint solvers. The frequent use of string
operations in these applications has motivated several groups to
explore the possibility of designing a constraint solver which treats
strings as first-class types. Such a specialized solver for strings
would further facilitate the use of constraint solving for analysis of
security applications with string operations.  

Software applications use various string predicates and functions
which are often made available to the developers as libraries. A
satisfiability solver for string constraints must be able to handle
these predicates and functions. From the string constraints and
predicates available in high level programming languages such as C,
JAVA and C++, we identify a set of core predicates and functions. Many
other more complicated string-manipulating functions can be expressed
as some simple composition of these functions. We use these predicates
and functions to define a theory of strings.  

The main contribution of this paper is an analysis of the complexity of several
fragments of the theory of strings. We show that fairly small and
simple-looking fragments are NP-complete. In light of the progress in SAT
solving and SMT solving for bit-vector arithmetic, these results indicate
that a SAT-based approach is reasonable to satisfiability solving of string
constraints.


\input{rel}

\input{theorydef}

\input{compl}



\section{Conclusion and Future Work}

The analysis of different fragments of the theory of strings presented in this paper 
shows that the satisfiability problem for even small non-trivial fragments
is NP-complete. Thus, it is unlikely that an efficient (polynomial-time)
algorithm for checking the satisfiability of the strings would be found. 
Hence, a simple approach based on Boolean encoding of string constraints 
to propositional logic is, in principle, as effective as any other technique  
for solving string constraints.
This justifies a ``byte-blast'' approach to solving string 
constraints which relies on encoding strings as bit-vectors and using 
an off-the-shelf bit-vector SMT solver.  Further, these hardness results 
underline the importance of using domain knowledge about string constraints 
arising out of security applications. We believe, in practice, 
word-level reasoning over strings that exploits such domain knowledge 
through pragmatic approaches such as abstraction-refinement might 
prove to be very effective in making an efficient and scalable for 
theory of strings. The key challenge in developing an SMT solver for 
theory of strings is identification of such properties of string 
constraints arising from real code.

Inspired by the success of abstraction-refinement based approaches for
SMT solving (e.g.,~\cite{kroening-cav04,bryant-tacas07,ganesh-cav07}),
we believe such an approach would be useful for the theory of
strings also.
We identify the abstraction techniques that we believe would be
especially useful in the context of a theory of strings:
\begin{mylist} 
\item {\em Length abstraction:} 
To our knowledge, this approach has been first published by 
Bjorner et al~\cite{bjorner-tr08}. It operates
by creating an over-approximation of the actual formula by
abstracting each string constraint with a corresponding length constraint. The
resulting integer linear arithmetic 
formula is solved to obtain candidate lengths for the
strings in the original formula, with a possible refinement needed
if these candidate lengths turn out to be too small.
We believe that this general idea can be used but with some guidance
to the solver to not simply generate the smallest lengths.

\item {\em Position abstraction:} 
We have observed that, in the security applications of interest, 
\emph{string-containment} is a widely used predicate and the
encoding the choice of position of containment adds significant
complexity to 
the constraint satisfaction problem. For large string-lengths, 
a standard byte-blast approach which
reduces the string constraints to bit-vector formula would require the SAT
solver to branch over a large set of choices of positions. We hypothesize based on
our observations of string constraints generated by
colleagues in security applications~\cite{juan-tr09}, 
that the position and order of containment of substrings is often not critical to finding
a satisfying assigment. Hence, an effective approach to construct
under-approximation of the string formula would be fixing some heuristic ordering
of containment constraints. If the formula with this fixed ordering is
unsatisfiable, the unsat core generated by the SAT solver can be used
to selectively refine the ordering.

\end{mylist}
The overall approach we envisage will be similar to the 
iterative construction of over- and under-approximate formulas as
performed in prior work on model checking~\cite{mcmillan-tacas03} and
SMT solving for bit-vector arithmetic~\cite{bryant-tacas07}.
It would be interesting to evaluate how such an approach based on abstraction-refinement
performs for string formulas generated in practice from security applications.

\section*{Acknowledgments}

We are grateful to Juan Caballero and Dawn Song for many helpful
discussions on formalizing fragments of the theory of strings that are
relevant for software security.

\bibliographystyle{abbrv}   
\bibliography{strings}       

\end{document}

%% file: macros.tex
\newcommand{\comment}[1]{}

\newcommand{\mycaption}[2]{\caption[#1]{{\bf #1} #2}}

\newtheorem{theorem}{Theorem}

\newtheorem{corollary}{Corollary}

\newcommand{\qed}{\mbox{} $\Box$}
\newenvironment{proof}{{\it Proof:}}{\qed}

\newcommand{\Naturals}{\ensuremath{\mathbb{N}}}

\newcommand{\bnf}{\mathrel{::=}}
\newcommand{\kw}[1]{{\bf #1}}
\newcommand{\va}[1]{{\rm {\it #1}}}
\newcommand{\mva}[1]{\mbox{\va{#1}}}
\newcommand{\syntaxform}[1]{\mva{#1}}
\newcommand{\formula}{\syntaxform{formula}}

\newcommand{\true}{\kw{true}}
\newcommand{\false}{\kw{false}}

\newcommand{\vbar}{\mathbin{|}}
\newcommand{\lvbar}{|\;}

\newcounter{myctr}
\newenvironment{mylist}{\begin{list}{\arabic{myctr}.}
{\usecounter{myctr}
\setlength{\topsep}{1mm}\setlength{\itemsep}{0.25mm}
\setlength{\parsep}{0.1mm}
\setlength{\itemindent}{0mm}\setlength{\partopsep}{0mm}
\setlength{\labelwidth}{15mm}
\setlength{\leftmargin}{4mm}}}{\end{list}}

\newcommand{\str}{\ensuremath{s}}
\newcommand{\strConst}{\ensuremath{c}}
\newcommand{\strContains}{\ensuremath{\sqsupseteq}}
\newcommand{\strContainsAt}[1]{\ensuremath{\strContains_{#1}}}
\newcommand{\concat}{\ensuremath{@}}

\newcommand{\strExp}{\syntaxform{str-expr}}
\newcommand{\boolExp}{\syntaxform{bool-expr}}


%% file: rel.tex
\section{Related Work}

Constraint solvers are widely used in verification and validation of software and hardware
systems~\cite{bryant-tacas07,gulwani-pldi08,tiwari-cade07}. In particular, they have been used
extensively for both static~\cite{mihai-issp05,brumley-csf07,brumley-usenix07} 
and dynamic analysis~\cite{egele-usenix07,godefroid-ndss08} of programs to detect malcious 
code or
 security
vulnerabilities in benign code. 
The use of constraint solving in software verification is driven by
development of faster and more scalable SMT solvers 
for the theory of bit-vectors such as
BAT~\cite{manolios-cav07}, 
Boolector~\cite{brummayer-tacas09}, Beaver~\cite{beaver-link}
MathSat~\cite{bruttomesso-cav07},
 Spear~\cite{hutter-fmcad07}, STP~\cite{ganesh-cav07},
 UCLID~\cite{bryant-tacas07} and Z3~\cite{z3}. In particular,
UCLID and STP have been successfully used for security applications. 
For example,
bit-vector solvers can be used to easily detect
overflow/underflow errors which are cause of many
security vulnerabilities such as buffer overflow~\cite{molnar-catchconvtr07}. 

Analysis of string processing software is an important 
problem~\cite{wassermann-pldi07,ruan-tase08,wassermann-issta08}.
This makes it essential to develop verification techniques
that can efficiently handle constraints over strings.
 A scalable approach for solving
string constraints must treat strings as first class types and 
string library functions
as native operations of the theory of strings~\cite{juan-tr09}.
Development of such a solver for a theory
of strings would further facilitate the use of constraint solving 
for program analysis, in general,
and security applications, in particular. This will further push the 
frontier of program
analysis in terms of scalability as well as program complexity.

While previous efforts have been made to develop decision 
procedures for regular expression
containment~\cite{hooimeijer-pldi09,christensen-SAS03}, there have been
some recent efforts to develop
 an SMT solver for the theory of
strings. 

In an independent and parallel work,
Kiezun et al~\cite{hampi} 
have developed a solver (HAMPI) for a theory of strings.
HAMPI works by reducing the formulae over string constraints to
bit-vector logic and then, using a bit-vector solver (STP) for
checking the satisfiability of the formulae. This reduction
is achieved in two steps. HAMPI reduces the string constraints
specified using a rich input language to a core theory of strings
comprising of regular language operations and membership predicate.
The string constraints in this core language are then translated
to bit-vector logic before invoking a bit-vector solver. They also
show that the satisfiability problem for this theory of strings with
regular expression operations is NP-complete.

The string theory considered in this paper is different from the
one considered in HAMPI. The string functions and predicates in our
theory of strings are motivated by commonly used library functions in
high level languages such as C and Java. The set of constraints expressible
in our theory of strings are not comparable with the set of constraints
expressible in HAMPI. We identify constraints which can be expressed in
our theory and not in HAMPI as well as those which can be expressed in HAMPI
but not in our theory. 
\begin{mylist}
\item Our theory has \emph{contains-at-position-i} predicate which is true 
if and
only if its first argument string is contained in its second argument string
at exactly position i. We also have \emph{extract-i-j} function which 
extracts a sub-string from its string argument using the indices i and j.
While the SMT solving approach of HAMPI can be used to handle these
constraints, the theory of strings considerd in HAMPI is based on regular
languages and can not be used to encode these constraints.
\item Our theory does not have union or star operation and hence, constraints
with union or star can not be expressed in our theory.
\end{mylist}

In particular, we note that the NP-completeness result established in 
Kiezun~\cite{hampi} relies on the use of union operation to provide
disjunction. \emph{We show that even without this union operation, the theory
of strings is NP-complete and all non-trivial fragments of the theory of
strings are also NP-complete.}

Bjorner et al~\cite{bjorner-tr08}
 propose another approach to solving string constraints
arising out of path feasibility queries. Their approach relies on identifying
candidate string lengths by solving length constraints and then, solving the
string constraints by considering them to be of lengths found in the first
phase. The string lengths found in the first phase may not provide a solution
even if the formulae is satisfiable and hence, they need to iterate with
different length assignments. The string operations considered in this
work are similar to ones proposed here. We consider strings of bounded lengths
and we do not consider the \emph{replace} operation. Hence, our fragment
of string theory is decidable. \emph{In contrast to Bjorner et al's work who
presented decidability result for theory of strings, we present
complexity results on the theory of strings and its different fragments.}

%% file: theorydef.tex
\section{Theory Definition}
\label{sec:theorydef}

The definition of the theory of strings presented in this section is motivated
by checking path feasibility queries over programs written in some high
level languages such as Java, C and Ocaml. The string libraries used in these
high level languages are abstracted as string functions and predicates.
We now define the complete theory of strings using these predicates and functions in this section. Later, we will analyze the complexity of different fragments of this theory by considering different subset of string predicates and functions. 

\begin{figure}[h]
\begin{eqnarray*}
\strExp & \bnf  & \strConst \vbar \str \vbar \strExp [i:j] \vbar \strExp_1 \concat \strExp_2 \\
&&  \\
\boolExp & \bnf & \true \vbar \false \vbar \neg \boolExp \\
&& \lvbar \strExp_1 = \strExp_2  \vbar  \strExp_1 \strContains \strExp_2 \vbar \strExp_1 \strContainsAt{i} \strExp_2 \\
&& \\
\formula & \bnf & \boolExp \vbar \boolExp \land \formula \\
&& \\
&& i,j \in \Naturals \hspace*{5mm} \str, \str_i \text{~are string variables~} \hspace*{5mm} \strConst \text{~represents a string constant.} \\
\end{eqnarray*}
\mycaption{Syntax for String Logic}{$[i:j]$ denotes extraction of the sub-string starting at position $i$ and
ending at position $j$; $\concat$ denotes concatenation; $\strContains$ 
denotes containment; and 
$\strContainsAt{i}$ denotes containment at position $i$.}
\label{fig:gram}
\end{figure}

The syntax of the statements in theory of strings is given by grammar in Figure~\ref{fig:gram}. The strings are over some finite alphabet $\Sigma$. The string constraints arising from software verification involve only finite length strings. The length of a string is bounded by the length of the corresponding buffer. So, we require that the maximum length of each string is bounded by a constant. Also, the maximum length of all
strings are less than some constant $L_{max}$. Also, there is an empty string constant $\epsilon$. We describe the semantics of the predicates and functions used in the theory definition below.


\textbf{String Predicates:} The string predicates take two string arguments and evaluate to true or false.

\begin{mylist}

\item Equality: $s_i = s_j$ is true if and only if both $s_i$ and $s_j$ are assigned the same string constants and otherwise, it is false.

\item Containment at position $i$: $s_1 \strContainsAt{i} s_2$ denotes that $s_2$ is contained in $s_1$ at position $i$. For example, $bombay$ contains $bay$ at location $4$. So, $bombay \strContainsAt{4} bay$ evaluates to true.

\item Containment: $s_1 \strContains s_2$ denotes that $s_2$ is contained in $s_1$ at some position. In particular, the empty string $\epsilon$ is contained in all strings and does not contain any non-empty string, that is, $\forall s, s \strContains \epsilon$ and $\forall s, s \not = \epsilon \Rightarrow \neg (\epsilon \strContains s)$.

\end{mylist}


\textbf{String Functions:} The two string functions considered in this paper are extraction and concatenation. 

\begin{mylist}

\item Extraction: $s [i:j]$ has the type signature $\strExp \times int \times int \rightarrow \strExp $. It denotes the substring of $s$ starting from position $i$ and ending at position $j$ where $i$ and $j$ are integers. For example, $bombay [4:6]$ evaluates to $bay$.\\

\item Concatenation: $s_1 \concat s_2$ has the type signature $\strExp \times \strExp \rightarrow \strExp$. It denotes the concatenation of the two strings provided to it as arguments. For example, $bom \concat bay$ evaluates to $bombay$.  \\

\end{mylist}

\vspace{-1cm}

%% file: compl.tex
\section{Complexity Results}

Before
stating and proving the complexity results, we present a brief summary
of the results in this section and note that all non-trivial fragments 
of theory of strings are NP-complete. In Theorem~\ref{th-np}, we 
show that the satisfiability problem for the theory of strings as define in
Section~\ref{sec:theorydef} is in NP. 
 This is a direct consequence of having a constant bound on the size of any string.
Hence, the 
satisfiability of any fragment of string theory is also in NP. Each fragment of theory 
of strings is defined by selecting some string predicates and functions along with Boolean 
negation and conjunction. As discussed in Section~\ref{sec:theorydef}, there are two 
functions and three predicates. To define a fragment of string theory we need to include 
atleast one string predicate. 

The three most elementary fragment of string theory are defined by including exactly one 
string predicate.
\begin{mylist}
\item E: This fragment consists of string equality, Boolean negation and conjunction.
\item C: This fragment consists of string containment, Boolean negation and conjunction.
\item T: This fragment consists of string containment at position $i$, Boolean negation 
and conjunction.
\end{mylist}
It is shown that the satisfiability problem for C fragment is NP-complete in 
Theorem~\ref{th-cont}. The satisfiability problem for T fragment is also NP-complete as
shown in Lemma~\ref{th-cat}. 
We know that E fragment is polynomial-time solvable using congruence
closure~\cite{barrett-smtbookch09}. 
So, we extend the E (equality) fragment with different string predicates and functions, 
and analyze its complexity. 
\begin{mylist}
\item E+C: This fragment extends E with string containment. 
\item E+T: This fragment extends E with string containment at position $i$. 
If $i$ is only allowed to be constant, the corresponding logic is E+T-CONST.
\item E+A: This fragment extends E with string concat function.
\item E+X: This fragment extends E with string extract function. 
If the indices for extract are only constant, the corresponding logic is E+X-CONST.
\end{mylist}

Since the satisfiability problems for C and T fragments are NP-complete, it 
is natural that E+C and E+T would also be NP-complete. We have separately 
proved the hardness results for both 
fragments in Theorem~\ref{th-eqcont} and Theorem~\ref{th-eqcat}. 
It is also shown in Theorem~\ref{th-eqconc} and Theorem~\ref{th-eqext} 
that the satisfiability problem for E+A and E+X fragments  are also NP complete.

Any extension of these fragments would also be NP-complete. So, the NP-completeness 
results for these minimal fragments of string theory presented in this section imply 
that the satisfiability problem for all fragments of string theory except for the E (equality) 
fragment is computationally 
hard. Thus, it is unlikely that there is any polynomial time algorithm for deciding the 
satisfiability of any non-trivial fragment of string theory unless P=NP.

\noop{
\begin{mylist}
\item equality and containment - EQCONT fragment of string theory
\item equality and containment at constant position - EQCAT-CONST fragment of 
string theory
\item equality and concatenation - EQCONC fragment of string theory
\item equality and extraction with constant indices - EQEXT-CONST fragment of 
string theory 
\end{mylist}
}

In the rest of the section, we state and prove the complexity results.

\begin{theorem} \label{th-np}
The satisfiability problem over the theory of string is in NP.
\end{theorem}

\begin{proof}
If the formula over theory of strings is satisfiable, then the satisfying instance is an assignment of string variables to strings with lengths upper bounded by the constant $L_{max}$. Hence, the size of the satisfying assignment is at most $L_{max}N$ where $N$ is the number of string variables. So, the length of the certificate is polynomial in the size of the input and hence, satisfiability of formula in theory of strings is in NP.
\end{proof}

\vspace{0.5cm}

As a consequence of the above theorem, the satisfiability of formulae in 
smaller fragments of theory of strings such as  E+C, E+T-CONST, E+A and E+X-CONST 
is also in NP. Hence, we only require to show that satisfiability of formulae in these 
fragments is NP-hard in order to prove that the satisfiability problem for these fragments 
is NP-complete. 

In rest of the section, we state and prove the NP-hardness results for each of these fragments. We show that the satisfiability problem for different fragments of string theory is NP-hard. Let us consider a 3-CNF formula $\phi$ over a set $X = \{x_1,x_2,\ldots,x_n\}$
 of $n$ Boolean variables.
$$ \phi \equiv \displaystyle \bigwedge_{i=1}^m (l_1^i \vee l_2^i \vee l_3^i)$$
where each literal $l^j_i$ is $x_k$ or $\neg x_k$ for some $x_k \in X$. We know that 3-CNF-SAT is NP-complete~\cite{cook-stoc71}. We now reduce this problem,
that is, finding an assignment of variables in
$X$ to $\{0,1\}$ such that $\phi$ evaluates to $1$, to the problem of finding
a satisfying assignment in the corresponding fragment of theory of strings.\\

\subsection{Equality + Containment (E+C)}

\begin{theorem} \label{th-eqcont}

The satisfiability problem over the theory of strings
with equality, contains, Boolean negation and conjunction (E+C fragment) 
is NP-hard.

\end{theorem}

\begin{proof}
We prove this by reducing 3-CNF-SAT to E+C fragment of theory of strings. We describe a transformation that maps a 3-CNF Boolean formula to
a formula in E+C fragment of theory of strings (over the alphabet $\Sigma = \{a\}$) such that there is a satisfying
assignment for the Boolean formula if and only if there is a satisfying
assignment for the formula over strings.\\

Let $\psi$ be defined as 
$$ \psi(x_i) \triangleq s_i \; and \; \psi(\neg x_i) \triangleq r_i$$ 
where $s_i$ and $r_i$ are strings of atmost length $1$.
$s_i = a$ if and only if $x_i$ is assigned true, otherwise, it is $\epsilon$.
Similarly, $r_i = a$ if and only if $x_i$ is assigned false,
otherwise it is $\epsilon$. So, for any literal $l$, $\psi(l)$ would be
$a$ if and only if $l$ is assigned true.


We also need to add
constraints to ensure consistency, that is, exactly one of $x_i$ or
$\neg x_i$ is assigned true. For consistency, for each variable
$x_i$, we must have the constraint
$$s_i \not = r_i$$
This ensures that exactly one of $s_i$ or $r_i$ is $a$.

Each clause $c \equiv l_1 \vee l_2 \vee l_3$ is transformed to
$$V_c \strContains \psi(l_1)$$
$$V_c \strContains \psi(l_2)$$
$$V_c \strContains \psi(l_3)$$
$$V_c \not = \epsilon$$
where $V_c$ is a new string variable for clause c and is of length atmost $3$.

Thus, atleast one of $\psi(l_i)$ must
be $a$ which is possible if and only $l_i$ is assigned true. So,
atleast one literal in each clause is true.

A set of string constraints $\psi(\phi)$ 
is obtained by applying the above transformations
to each clause $c_i$ in 3-CNF Boolean formula $\phi$ 
and taking the union of all the obtained string constraints.

Let $\mathcal I$ be a satisfying assignment to $\phi$ such that
$\mathcal I(x)$ denotes the assignment to $x$. By construction, there is an
assignment $\mathcal I'$ to $\psi(\phi)$ such that $\mathcal I'(s_i) = a$ and
$\mathcal I'(r_i) = \epsilon$ if and only if $\mathcal I(x_i) = true$.

Thus, E+C fragment of string theory is NP-hard.
\end{proof}

\begin{corollary}

The satisfiability problem over the theory of strings
with equality, contains at $i$ where $i$ is variable, Boolean negation and conjunction (E+T-VAR fragment) is NP-hard.

\end{corollary}

\begin{proof}
$str_1 \strContains str_2$ can be rewritten as $str_1 \strContainsAt i str_2$ where $i$ is a new index variable. Hence, any formula in E+C can be expressed
as a formula in E+T-VAR fragment.
\end{proof}

\subsection{Equality + Containment-aT-Constant (E+T-CONST)}

\begin{theorem} \label{th-eqcat}

The satisfiability problem over the theory of strings
with contains at constant position, equality, Boolean negation and conjunction (E+T-CONST fragment) 
is NP-hard.

\end{theorem}

\begin{proof}
We prove this by reducing 3-CNF-SAT to E+T-CONST fragment of theory of strings. We describe a transformation that maps a 3-CNF Boolean formula to
a formula in E+T-CONST fragment of theory of strings (over the alphabet $\Sigma = \{a,b\}$) such that there is a satisfying
assignment for the Boolean formula if and only if there is a satisfying
assignment for the formula over strings.\\

Let $\psi$ be defined as 
$$ \psi(x_i) \triangleq s_i \; and \; \psi(\neg x_i) \triangleq r_i$$ 
where $s_i$ and $r_i$ are strings of atmost length $1$. To make it
exactly of length $1$, we  require
$s_i \not = \epsilon \wedge r_i \not = \epsilon$.
$s_i = a$ if and only if $x_i$ is assigned true, otherwise, it is $b$.
Similarly, $r_i = a$ if and only if $x_i$ is assigned false,
otherwise it is $b$. So, for any literal $l$, $\psi(l)$ would be
$a$ if and only if $l$ is assigned true.


We also need to add
constraints to ensure consistency, that is, exactly one of $x_i$ or
$\neg x_i$ is assigned true. For consistency, for each variable
$x_i$, we must have the constraint
$$s_i \not = r_i$$
This ensures that exactly one of $s_i$ or $r_i$ is $a$.

Each clause $c \equiv l_1 \vee l_2 \vee l_3$ is transformed to
$$V_c \strContainsAt 1 \psi(l_1)$$
$$V_c \strContainsAt 2 \psi(l_2)$$
$$V_c \strContainsAt 3 \psi(l_3)$$
$$V_c \not = bbb$$
where $V_c$ is a new variable for clause c and is of length atmost $3$.

Thus, atleast one of $\psi(l_i)$ must
be of $a$ which is possible if and only if $l_i$ is assigned true. So,
atleast one literal in each clause is true.

A set of string constraints $\psi(\phi)$ 
is obtained by applying the above transformations
to each clause $c_i$ in 3-CNF Boolean formula $\phi$ 
and taking the union of all the obtained string constraints.

Let $\mathcal I$ be a satisfying assignment to $\phi$ such that
$\mathcal I(x)$ denotes the assignment to $x$. By construction, there is an
assignment $\mathcal I'$ to $\psi(\phi)$ such that $\mathcal I'(s_i) = a$ and
$\mathcal I'(r_i) = b$ if and only if $\mathcal I(x_i) = true$.

Thus, E+T-CONST fragment of string theory is NP-hard.
\end{proof}

\begin{corollary} \label{th-cat}

The satisfiability problem over the theory of strings
with contains at constant position, Boolean negation and conjunction 
is NP-hard.

\end{corollary}

\begin{proof}
In the proof above, we can replace $s_i \not = r_i$ 
by $\neg(s_i \strContainsAt 1 r_i)$ and $V_c \not = bbb$ by
$\neg(V_c \strContainsAt 1 bbbb)$. The NP-hardness proof still 
goes through. Dis-equality
between the strings of same length
can be expressed as dis-containment-at-1.
\end{proof}

\subsection{Equality + concAt (E+A)}

\begin{theorem} \label{th-eqconc}

The satisfiability problem over the theory of strings
with equality, concat and Boolean conjunction (E+A fragment) is NP-hard.

\end{theorem}

\begin{proof}
We prove this by reducing 3-CNF-SAT to E+A fragment of theory of strings.
We describe a transformation that maps a 3-CNF Boolean formula to
a formula in E+A fragment of theory of strings 
(over the alphabet $\Sigma = \{a\}$) such that there is a satisfying
assignment for the Boolean formula if and only if there is a satisfying
assignment for the formula over strings.\\
Let $\psi$ be defined as 
$$ \psi(x_i) \triangleq s_i \; and \; \psi(\neg x_i) \triangleq r_i$$ 
where $s_i$ and $r_i$ are strings of atmost length $1$.
$s_i = a$ if and only if $x_i$ is assigned true, otherwise, it is $\epsilon$.
Similarly, $r_i = a$ if and only if $x_i$ is assigned false,
otherwise it is $\epsilon$. So, for any literal $l$, $\psi(l)$ would be
$a$ if and only if $l$ is assigned true. 

We also need to add
constraints to ensure consistency, that is, exactly one of $x_i$ or
$\neg x_i$ is assigned true. For consistency, for each variable
$x_i$, we must have the constraint
$$s_i \concat r_i = a$$
This ensures that exactly one of $s_i$ or $r_i$ is $a$, that is,
exactly one of $\psi(x_i)$ or $\psi(\neg x_i)$ is $a$.

Each clause $l_1 \vee l_2 \vee l_3$ is transformed to
$$\psi(l_1) \concat \psi(l_2) \concat \psi(l_3)   \concat p_i =
aaa$$ 
where $p_i$ is of length atmost $2$.
Thus, the sum of the lengths of $\psi(l_1)$, $\psi(l_2)$ and
$\psi(l_3)$ must be atleast 1, that is, atleast one of $\psi(l_i)$ must
be $a$ which is possible if and only $l_i$ 
is assigned true. So,
atleast one literal in each clause is true.

A set of string constraints $\psi(\phi)$ 
is obtained by applying the above transformations
to each clause $c_i$ in 3-CNF Boolean formula $\phi$ 
and taking the union of all the obtained string constraints.

Let $\mathcal I$ be a satisfying assignment to $\phi$ such that
$\mathcal I(x)$ denotes the assignment to $x$. By construction, there is an
assignment $\mathcal I'$ to $\psi(\phi)$ such that $\mathcal I'(s_i) = a$ and
$\mathcal I'(r_i) = \epsilon$ if and only if $\mathcal I(x_i) = true$.

Thus, E+A fragment of string theory is NP-hard.

\end{proof}

\noop{
As a corollary of the above theorem, we obtain the following
result.

\begin{corollary}

The satisfiability problem over the theory of strings with equality and
concat (without Boolean conjunction) is NP-hard.

\end{corollary}

\begin{proof}

Introduce a new letter $b$ in the alphabet $\Sigma$ and use concat with 
the new letter $b$ as separators to replace conjunctions. 
$str_1 = str_2 \wedge str_3 = str_4$ can be replaced by a single constraint 
$str_1 \concat b \concat str_3 = str_2 \concat b \concat str_4$ where $b$ 
is a new letter. This is possible because we do not have any negation in
this fragment.


\end{proof}
}

\begin{corollary}

The satisfiability problem over the theory of strings
with contains ($\strContains$) and concat (C+A fragment) is NP-hard.

\end{corollary} 

\begin{proof}
Equality can be expressed with two-way containment.
Once again, note that there is no negation in this fragment.
\end{proof}

\subsection{Equality + eXtract-with-constant-indices (E+X-Const)}

\begin{theorem} \label{th-eqext}

The satisfiability problem over the theory of strings
with equality, extract with constant indices, Boolean negation and conjunction (E+X-CONST fragment) is NP-hard.

\end{theorem}

\begin{proof}
We prove this by reducing 3-CNF-SAT to E+X-CONST fragment of theory of strings. We describe a transformation that maps a 3-CNF Boolean formula to
a formula in E+X-CONST fragment of theory of strings (over the alphabet $\Sigma = \{a,b\}$) such that there is a satisfying
assignment for the Boolean formula if and only if there 
is a satisfying
assignment for the formula over strings.\\
Let $\psi$ be defined as 
$$ \psi(x_i) \triangleq s_i \; and \; \psi(\neg x_i) \triangleq r_i$$ 
where $s_i$ and $r_i$ are strings of atmost length $1$. To make it
exactly of length $1$, we  require
$s_i \not = \epsilon \wedge r_i \not = \epsilon$.
$s_i = a$ if and only if $x_i$ is assigned true, otherwise, it is $b$.
Similarly, $r_i = a$ if and only if $x_i$ is assigned false,
otherwise it is $b$. So, for any literal $l$, $\psi(l)$ would be
$a$ if and only if $l$ is assigned true. 


We also need to add
constraints to ensure consistency, that is, exactly one of $x_i$ or
$\neg x_i$ is assigned true. For consistency, for each variable
$x_i$, we must have the constraint
$$s_i \not = r_i$$
This ensures that exactly one of $s_i$ or $r_i$ is $a$.

Each clause $c \equiv l_1 \vee l_2 \vee l_3$ is transformed to
$$V_c[1:1] = \psi(l_1)$$
$$V_c[2:2] = \psi(l_2)$$
$$V_c[3:3] = \psi(l_3)$$
$$V_c \not = bbb$$
where $V_c$ is a new variable for clause c and is of length atmost $3$.

Thus, atleast one of $\psi(l_i)$ must
be of $a$ which is possible if and only $l_i$ is assigned true. So,
atleast one literal in each clause is true.

A set of string constraints $\psi(\phi)$ 
is obtained by applying the above transformations
to each clause $c_i$ in 3-CNF Boolean formula $\phi$ 
and taking the union of all the obtained string constraints.

Let $\mathcal I$ be a satisfying assignment to $\phi$ such that
$\mathcal I(x)$ denotes the assignment to $x$. By construction, there is an
assignment $\mathcal I'$ to $\psi(\phi)$ such that $\mathcal I'(s_i) = a$ and
$\mathcal I'(r_i) = b$ if and only if $\mathcal I(x_i) = true$.

Thus, E+X-CONST fragment of string theory is NP-hard.


\end{proof}

We now show that even without equality, the fragment of the theory of strings having \emph{contains} as string predicate with Boolean negation and conjunction is also hard. This is the final result of the section.

\subsection{Containment (C)}

\begin{theorem} \label{th-cont}

The satisfiability problem over the theory of strings
with contains, Boolean negation and conjunction (C fragment) is NP-hard.

\end{theorem}

\begin{proof}
We prove this by reducing 3-CNF-SAT to C fragment of theory of strings.
We describe a transformation that maps a 3-CNF Boolean formula to
a formula in C fragment of theory of strings (over the alphabet $\Sigma = \{a,b\}$) such that there is a satisfying
assignment for the Boolean formula if and only if there is a satisfying
assignment for the formula over strings.\\
Let $\psi$ be defined as 
$$ \psi(x_i) \triangleq s_i \; and \; \psi(\neg x_i) \triangleq r_i$$ 
where $s_i$ and $r_i$ are strings of atmost length $1$. To make it
exactly of length $1$, we  require
$\neg (\epsilon \strContains s_i) \wedge \neg (\epsilon \strContains r_i)$.
$s_i = a$ if and only if $x_i$ is assigned true, otherwise, it is $b$.
Similarly, $r_i = a$ if and only if $x_i$ is assigned false,
otherwise it is $b$. So, for any literal $l$, $\psi(l)$ would be
$a$ if and only if $l$ is assigned true. 


We also need to add
constraints to ensure consistency, that is, exactly one of $x_i$ or
$\neg x_i$ is assigned true. For consistency, for each variable
$x_i$, we must have the constraint
$$\neg(s_i \strContains r_i)$$

Each clause $c \equiv l_1 \vee l_2 \vee l_3$ is transformed to
$$V_c \strContains \psi(l_1)$$
$$V_c \strContains \psi(l_2)$$
$$V_c \strContains \psi(l_3)$$
$$\neg (bbb \strContains V_c)$$
where $V_c$ is a new variable for clause c and is of length atmost $3$.

Thus, atleast one of $\psi(l_i)$ must
be of $a$ which is possible if and only $l_i$ is assigned true. So,
atleast one literal in each clause is true.

A set of string constraints $\psi(\phi)$ 
is obtained by applying the above transformations
to each clause $c_i$ in 3-CNF Boolean formula $\phi$ 
and taking the union of all the obtained string constraints.

Let $\mathcal I$ be a satisfying assignment to $\phi$ such that
$\mathcal I(x)$ denotes the assignment to $x$. By construction, there is an
assignment $\mathcal I'$ to $\psi(\phi)$ such that $\mathcal I'(s_i) = a$ and
$\mathcal I'(r_i) = b$ if and only if $\mathcal I(x_i) = true$.

Thus, C fragment of string theory is NP-hard.
\end{proof}

